\documentclass[12pt]{amsart}

\usepackage[centertags]{amsmath}
\usepackage{amsthm,amsfonts}
\usepackage{mathabx} 
\usepackage{tikz, float}
\usepackage[toc,page]{appendix}

\usepackage{bm}

\usepackage{amssymb}
\usepackage{epsfig}
\usepackage{amssymb, latexsym, dsfont, mathrsfs}
\usepackage{indentfirst}

\usepackage{color}
\usepackage{ulem}
\usepackage{verbatim}

\begin{document}
\theoremstyle{plain}
\newtheorem{theorem}{Theorem}[section]
\newtheorem{lemma}{Lemma}[section]
\newtheorem{proposition}{Proposition}[section]
\newtheorem{corollary}{Corollary}[section]
\newtheorem{example}{Example}[section]
\theoremstyle{definition}
\newtheorem{notations}[theorem]{Notations}
\newtheorem{notation}[theorem]{Notation}
\newtheorem{remark}[theorem]{Remark}
\newtheorem{question}[theorem]{Question}
\newtheorem{definition}[theorem]{Definition}
\newtheorem{condition}[theorem]{Condition}
\newtheorem{open problem}[theorem]{Open Problem}
\newcommand{\Hom}{{\rm Hom}}
\newcommand{\Ima}{{\rm Im}}
\newcommand{\Ker}{{\rm Ker \ }}
\newcommand{\rank}{{\rm rank}}
\newcommand{\Norm}{{\rm Norm}}
\newcommand{\rowspace}{{\rm rowspace}}
\newcommand{\Span}{{\rm Span}}
\newcommand{\Tr}{{\rm Tr}}
\newcommand{\Char}{{\rm char \ }}

\newcommand{\soplus}[1]{\stackrel{#1}{\oplus}}
\newcommand{\dlog}{{\rm dlog}\,}    

\newenvironment{pf}{\noindent\textbf{Proof.}\quad}{\hfill{$\Box$}}

\newcommand{\sA}{{\mathcal A}}
\newcommand{\sC}{{\mathcal C}}
\newcommand{\sD}{{\mathcal D}}
\newcommand{\sG}{{\mathcal G}}
\newcommand{\sH}{{\mathcal H}}
\newcommand{\sL}{{\mathcal L}}
\newcommand{\sP}{{\mathcal P}}

\newcommand{\F}{{\mathbb F}}
\newcommand{\Z}{{\mathbb Z}}
\newcommand{\R}{{\mathbb R}}
\newcommand{\N}{{\mathbb N}}
\newcommand{\A}{{\mathbb A}}
\newcommand{\K}{{\mathbb K}}
\newcommand{\C}{{\mathbb C}}
\newcommand{\x}{{\mathbf{x}}}
\newcommand{\be}{\begin{eqnarray}}
\newcommand{\ee}{\end{eqnarray}}
\newcommand{\nn}{{\nonumber}}
\newcommand{\dd}{\displaystyle}
\newcommand{\ra}{\rightarrow}

\newcommand{\SqBinom}[2]{\genfrac{[}{]}{0pt}{}{#1}{#2}}
\newcommand{\Binom}[2]{\genfrac{(}{)}{0pt}{}{#1}{#2}}

\newcommand{\Keywords}[1]{\par\noindent
{\small{\it Keywords\/}: #1}}

\title[Complete $b$-Symbol Weight Distribution of some Irreducible Cyclic Codes]{Complete $b$-Symbol Weight Distribution of Some Irreducible Cyclic Codes}

\author[ Hongwei Zhu \& Minjia Shi \&  Ferruh \"{O}zbudak 
]{Hongwei Zhu \& Minjia Shi \&  Ferruh \"{O}zbudak  
}
\thanks{MSC codes: 94B14, 11T71, 94B27}

\thanks{Hongwei Zhu is with School of Mathematical Sciences, Anhui University, Hefei, China. E-mail: zhwgood66@163.com.}

\thanks{Minjia Shi is with Key Laboratory of Intelligent Computing Signal Processing, Ministry of Education,  School of Mathematical Sciences, Anhui University, Hefei, Anhui, 230601, China, e-mail: smjwcl.good@163.com}

\thanks{Ferruh \"{O}zbudak is with Department of Mathematics and Institute of Applied Mathematics, Middle East Technical
	University,   Ankara, Turkey;
	e-mail: ozbudak@metu.edu.tr}


\abstract
Recently, $b$-symbol codes are proposed to protect against $b$-symbol errors in $b$-symbol read channels. It is an interesting subject of study  to consider the complete $b$-symbol weight distribution of cyclic codes since $b$-symbol metric is a generalization for Hamming metric. The complete $b$-symbol Hamming weight distribution of irreducible codes is known in only a few cases. In this paper, we give a complete $b$-symbol Hamming weight distribution of a class of irreducible codes with two nonzero $b$-symbol Hamming weights.

\vspace{0.7cm}
\Keywords{cyclic code, $b$-symbol error, $b$-symbol Hamming weight distribution,  irreducible cyclic code}
\endabstract
\maketitle

\section{Introduction}\label{sec:1}
In traditional information theory, people often analyze noise channels by dividing information into independent information units. However, with the development of storage technology, one finds that symbols can not always be written or read continuously. In 2011, Cassuto and Blaum \cite{CB1,CB2} first proposed a new coding framework for read channel based on symbol-pair. In this channel, the output of the reading process is the overlapping symbol-pair of the read channel. Symbol-pair codes are designed to correct symbol-pair errors.
Later, Cassuto and Litsyn \cite{CL} studied the symbol-pair codes corresponding to the cyclic codes and gave the lower bound for symbol-pair distance of cyclic codes through discrete Fourier transform and BCH bound.
Next, Chee {\it et al}. \cite{C+,C+1} established the Singleton-type bound of codes with symbol-pair metric, and considered the construction of symbol-pair codes reaching this bound. More related results on this topic can be found in \cite{CLL,DGZ,DTG,KZL}.
Recently, Yaakobi {\it et al}. \cite{YBS} generalized the symbol-pair read channel to the $b$-symbol read channel ($2$-symbol is exactly symbol-pair) and gave a decoding algorithm based on a bounded distance decoder for the cyclic code, where $b \geq 1$. When $b=1$, the $1$-symbol weight is exactly Hamming weight. Ding {\it et al.} \cite{DTG} established the Singleton-type bound for $b$-symbol codes, and constructed several infinite families of linear maximum distance separable (MDS) $b$-symbol codes. In the $b$-symbol read channel, one tends to consider the cyclic codes. Important motivations for considering cyclic codes include that cyclic shifting does not change the $b$-symbol weight of the codeword, there are good decoding algorithms and cyclic codes have good algebraic
 structure.
Hence it is natural to consider $b$-symbol weight distribution of cyclic codes. It is well known that the problem of determining the weight distribution of an irreducible cyclic code is notoriously difficult in general. Similar difficulty holds in the case of arbitrary cyclic codes and also for their $b$-symbol weights, in general. There are only few cases that the $b$-symbol weight distribution of cyclic codes are completely determined. Sun {\it et al.} \cite{Z} gave the $2$-symbol distance distribution of a class of repeated-root cyclic codes over finite fields. Using geometric approach, Shi {it et al.} \cite{SOS} obtain, among other results,  tight lower and upper bounds on $b$-symbol Hamming weight of arbitrary cyclic codes and a class of irreducible cyclic codes with constant $b$-symbol Hamming weight.

In this paper we consider a natural next step after \cite{SOS}. We completely determine the $b$-symbol weight distribution of irreducible cyclic codes of length $n$ over $\F_q$ if $\gcd\left( \frac{q^r-1}{q-1},\frac{q^r-1}{n}\right)=2$ with $n \mid \left(q^r-1\right)$ and $r \ge 2$. These results correspond to two-weight cyclic codes over the alphabet $\F_q \times  \cdots \times \F_q=\F_q^b$, which is not a field (see Remark \ref{remark.2.weight} below). It turns out that the complete determination of
$b$-symbol weights of these irreducible cyclic codes are more difficult and we obtain our results up to a new invariant $\mu(b)$ that we introduce in Definition \ref{definition.mu.b} below. We find that $\mu(b)$ is a natural arithmetic and geometric invariant and we determine it exactly if $b=r$. We also obtain numerical results for $2 \le b < r$.

Our main results are Theorems \ref{theorem.complete.b.symbol.C.b.small} and \ref{theorem.complete.b.symbol.C.b.large}. These results are stronger than the $b$-symbol weight distributions as we determine the $b$-symbol weight of any given codeword for the irreducible cyclic codes of this paper. Another important problem in this area is construction of maximum distance separable (MDS) $b$-symbol codes. There are interesting constructions in the recent literature using various methods, for example \cite{C+,C+1,DGZ,DTG,KZL}. As a consequence of Theorem \ref{theorem.complete.b.symbol.C.b.large} we observe that all of the irreducible cyclic codes we study here are MDS $b$-symbol if $b=r$.

This paper is organized as follows. In section 2, we introduce basic notations and definitions. In section 3, we show the main results. In section 4, we give the detailed proof of the main results. In the appendix, using exponent sum, we give another proof to show the equality (\ref{ep9.proof.Theorem.small.b}) holds.
\section{Preliminaries}\label{sec:preliminaries}

Throughout this paper we assume  and fix the following:

\begin{itemize}

\item $\F_q$: finite field with $q$ elements.  \\

\item $\F_q^*$: $\F_q\backslash\{0\}$.\\

\item $q=p^e$, $p=\Char \F_q$ and $p$ is odd. \\

\item $r \ge 2$: an {\it even} integer. \\

\item $nN=q^r-1$, where $n,N$ are positive integers. \\

\item $\gcd\left(\frac{q^r-1}{q-1},N\right)=2$. \\

\item $2 \le b \le n-1$ : an integer. \\

\item $\eta \in \F_{q^r}$: a primitive $(q^r-1)$-th root of $1$, or equivalently a primitive element of $\F_{q^r}$. \\
\item $\Tr:\F_{q^r} \ra \F_q$: the trace map defined as $x \mapsto x+ x^q + \cdots +x^{q^{r-1}}$.
\end{itemize}

Denote by $w(\mathbf{x})$ or $w_H(\mathbf{x})$ the Hamming weight of $\mathbf{x}\in \F_q^n.$
The $b$-symbol Hamming weight $w_b(\x)$ of $\x=(x_0,\ldots,x_{n-1})\in \F_q^n$ is defined as the Hamming weight of $\pi_b(\x)$, where
\be \label{definition.pi.b}
\pi_b(\x)=((x_0,\ldots,x_{b-1}),(x_{1},\ldots,x_{b}),
\cdots,(x_{n-1},\ldots,x_{b+n-2({\rm mod}~n)}))
\ee
is in $(\F_q^b)^n.$
When $b=1$, $w_1(\x)$ is exactly the Hamming weight of $\x$.
For any $\x,\mathbf{y}\in \F_q^N$, we have $\pi_b(\x+\mathbf{y})=\pi_b(\x)+\pi_b(\mathbf{y})$, and
the $b$-symbol distance ($b$-distance for short) $d_b(\x,\mathbf{y})$ between
$\x$ and $\mathbf{y}$ is defined as $d_b(\x,\mathbf{y})=w_b(\x-\mathbf{y}).$  Let $A_i^{(b)}$ denote the number of codewords with
 $b$-symbol Hamming weight $i$ in a code $C$ of length $n$. The $b$-symbol Hamming weight enumerator of $C$ is defined by
\begin{equation*}\label{b-enumerator}
  1+A_1^{(b)}T+A_2^{(b)}T^2+\cdots+A_n^{(b)}T^{n}.
\end{equation*}
Ding {\it et al}. \cite{DTG} established a Singleton-type bound for $b$-symbol codes.
Let $q\geq 2$ and $b\leq d_b(C)\leq n$. If $C$ is an $(n,M,d_b(C))_q$ $b$-symbol code, then we have $M\leq q^{n-d_b(C)+b}$.
An $(n,M,d_b(C))_q$ $b$-symbol code $C$ with $M=q^{n-d_b(C)+b}$ is called a maximum distance separable (MDS for short) $b$-symbol code.

For $a \in \F_{q^r}$, let $c(a) \in \F_q^{n}$ be the codeword defined as
\be
c(a)=\left( \Tr( a \eta^{0 \cdot N}), \Tr(a \eta^{1 \cdot N}), \ldots, \Tr(a \eta^{j \cdot N}), \ldots, \Tr(a \eta^{(n-1)\cdot N})\right),
\nn\ee
where $0 \le j \le n-1$.

Let $C$ be the $\F_q$-linear code of length $n$ defined as
\be
C= \left\{c(a): a \in \F_{q^r}\right\}.
\nn\ee

The following simple lemma is useful.

\begin{lemma} \label{lemma.C.irreducible}
$C$ is the irreducible code of length $n$ over $\F_q$ and $\dim_{\F_q} C=r$.
\end{lemma}

Next we introduce an important invariant $\mu(b)$ of the extension $\F_{q^r}/\F_q$. First we introduce
a subset $\mathcal{P}(b)$ of $\F_{q^r}^*$, which we use in the definition of $\mu(b)$ below.

\begin{definition} \label{definition.set.P.b}
Let $2 \le b \le r$ be an integer. It follows from \cite[Corollary 4.1]{SOS} that the set $\left\{1,\eta^N, \eta^{2N}, \ldots, \eta^{(b-1)N}\right\}$ is linearly independent over $\F_q$. Let $\mathcal{P}(b)$ be the subset of cardinality $(q^b-1)/(q-1)$ in $\F_{q^r}^*$ defined as
\be
\begin{array}{rcl}
\mathcal{P}(b) & = & \dd \bigcup_{j=1}^{b-1} \left\{\eta^{(j-1)N} + x_1 \eta^{jN} + \cdots + x_{b-j} \eta^{(b-1)N}: x_1, \cdots, x_j \in \F_q \right\}\dd \cup \left\{ \eta^{(b-1)N} \right\}.
\end{array}
\nn\ee
\end{definition}

\begin{example} For $r \ge 3$ we have
\be
\mathcal{P}(2)= \left\{1 + x_1 \eta^N: x_1 \in \F_q\right\} \cup \left\{\eta^N\right\}
\nn\ee
and
\be
\begin{array}{rcl}
\mathcal{P}(3) & = & \dd \left\{1 + x_1 \eta^{N} + x_2 \eta^{2N}: x_1, x_2 \in \F_q \right\} \dd \cup \left\{\eta^{N} + x_1 \eta^{2N}: x_1 \in \F_q \right\}  \dd \cup \left\{ \eta^{2N} \right\}.
\end{array}
\nn\ee
\end{example}
Now we are ready to define the invariant $\mu(b)$.

\begin{definition} \label{definition.mu.b}
For $2 \le b \le r$, let
\be
\mu(b)=\left|\left\{ x  \in \mathcal{P}(b): x \; \mbox{is a square in} \; \F_{q^r}^*\right\}\right|.
\nn\ee
\end{definition}

It seems difficult to determine $\mu(b)$ exactly when $b < r$. We determine it exactly when $b=r$ in the following.
\begin{lemma} \label{lemma.determine.mu.b.is.r}
Under notation as above, we have
\be
\mu(r)=\frac{1}{2}\left|\mathcal{P}(b)\right|=\frac{q^r-1}{2(q-1)}.
\nn\ee
\end{lemma}
\begin{proof}
First note that $(q^r-1)/(q-1)$ is an even integer as $r$ is even. Moreover any element of $\F_q^*$ is a square in $\F_{q^r}^*$ as $r$ is even. It follows from Definition \ref{definition.set.P.b} that
\be
\F_{q^r}^* = \bigsqcup_{x \in \mathcal{P}(b)} x \F_q^*,
\nn\ee
where $\bigsqcup$ is a disjoint union. Hence $x \in \mathcal{P}(b)$ is a square in $\F_{q^r}^*$ if and only if each element of $x \F_q^*$ is a square in $\F_{q^r}^*$. This implies that
\be
\mu(b)(q-1)=\frac{q^r-1}{2},
\nn\ee
which is the number of squares in $\F_{q^r}^*$. This completes the proof.
\end{proof}
\begin{example}\label{ep2.2}
Here are some numerical examples about $\mu(b)$ which computed by Magma.\\
\begin{center}
\begin{tabular}{l|c|c|c|c|c}
  \hline
  $p$ & $q$ & $r$ & $N$ & $b$ & $\mu(b)$ \\
  \hline
  $3$ & $3$ & $2$ & $2$ & $2$ & $2=\frac{q^r-1}{2(q-1)}$ \\
  $3$ & $3$ & $4$ & $2$ & $2$ & $3$ \\
  $3$ & $3$ & $4$ & $2$ & $3$ & $8$ \\
  $3$ & $3$ & $4$ & $2$ & $4$ & $20=\frac{q^r-1}{2(q-1)}$ \\
  $5$ & $5$ & $2$ & $2$ & $2$ & $3=\frac{q^r-1}{2(q-1)}$ \\
  $5$ & $5$ & $4$ & $2$ & $2$ & $4$ \\
  $5$ & $5$ & $4$ & $2$ & $3$ & $18$ \\
  $5$ & $5$ & $4$ & $2$ & $4$ & $78=\frac{q^r-1}{2(q-1)}$ \\
  $3$ & $9$ & $2$ & $2$ & $2$ & $5=\frac{q^r-1}{2(q-1)}$ \\
  $3$ & $9$ & $4$ & $2$ & $2$ & $4$ \\
  $3$ & $9$ & $4$ & $2$ & $3$ & $50$ \\
  $3$ & $9$ & $4$ & $2$ & $4$ & $410=\frac{q^r-1}{2(q-1)}$ \\
  $3$ & $9$ & $6$ & $2$ & $2$ & $4$\\
  $3$ & $9$ & $6$ & $2$ & $3$ & $51$\\
  $3$ & $9$ & $6$ & $2$ & $4$ & $401$\\
  $3$ & $9$ & $6$ & $2$ & $5$ & $3728$\\
  $3$ & $9$ & $6$ & $2$ & $6$ & $33215=\frac{q^r-1}{2(q-1)}$ \\
  $5$ & $25$ & $2$ & $2$ & $2$ & $13=\frac{q^r-1}{2(q-1)}$ \\
  $5$ & $25$ & $4$ & $2$ & $2$ & $11$ \\
  $5$ & $25$ & $4$ & $2$ & $3$ & $338$ \\
  $5$ & $25$ & $4$ & $2$ & $4$ & $8138=\frac{q^r-1}{2(q-1)}$ \\
  \hline
\end{tabular}
\end{center}

\end{example}

According to Example \ref{ep2.2}, when $2\leq b<r,$ we notice that the value of $\mu(b)$ is close to, but not equal to, $\frac{q^b-1}{2(q-1)}.$ It is natural for us to arise such an open question as follows.

\begin{open problem} \label{open problem}
Determine the invariant $\mu(b)$ when $2\leq b<r$ or give  {\it good} lower and upper bounds to $\mu(b)$.
\end{open problem}
\section{The main results}
Next we state our first main result. We completely determine the $b$-symbol Hamming weight of a given arbitrary nonzero codeword of $C$ for $2 \le b < r$. Recall that $q=p^e$, $p=\Char \F_q$ and $p$ is odd.

\begin{theorem} \label{theorem.complete.b.symbol.C.b.small}
Let $a \in \F_{q^r}^*$. Assume that $2 \le b < r$. Then we determine $w_b(c(a))$ explicitly as follows:
\begin{itemize}
\item If $p \equiv 1 \mod 4$ and $a$ is a square in $\F_{q^r}^*$,  then
\be
\begin{array}{rcl}
w_b(c(a)) & = & \dd \frac{q^b-1}{N(q-1)q^{b-1}} \left( q^r - \frac{q^r+(q-1)q^{r/2}}{q}\right)\dd +\frac{2 \mu(b) (q-1) q^{r/2}}{Nq^b}.
\end{array}
\nn\ee
\item If $p \equiv 1 \mod 4$ and $a$ is a non-square in $\F_{q^r}^*$, then
\be
\begin{array}{rcl}
w_b(c(a)) & = & \dd \frac{q^b-1}{N(q-1)q^{b-1}} \left( q^r - \frac{q^r-(q-1)q^{r/2}}{q}\right) \dd -\frac{2 \mu(b) (q-1) q^{r/2}}{Nq^b}.
\end{array}
\nn\ee
\item If $p \equiv 3 \mod 4$ and $a$ is a square in $\F_{q^r}^*$, then
\be
\begin{array}{rcl}
w_b(c(a)) & = & \dd \frac{q^b-1}{N(q-1)q^{b-1}} \left( q^r - \frac{q^r+(-1)^{er/2}(q-1)q^{r/2}}{q}\right) \dd +\frac{2 \mu(b) (-1)^{er/2}(q-1) q^{r/2}}{Nq^b}.
\end{array}
\nn\ee
\item If $p \equiv 3 \mod 4$ and $a$ is a non-square in $\F_{q^r}^*$, then
\be
\begin{array}{rcl}
w_b(c(a)) & = & \dd \frac{q^b-1}{N(q-1)q^{b-1}} \left( q^r - \frac{q^r-(-1)^{er/2}(q-1)q^{r/2}}{q}\right)  \dd -\frac{2 \mu(b) (-1)^{er/2}(q-1) q^{r/2}}{Nq^b}.
\end{array}
\nn\ee
\end{itemize}
\end{theorem}

\begin{remark} \label{remark.2.weight}
The Hamming weight distribution of the above cyclic codes has been considered in \cite{Ding}, and $C$ is a two-weight code  under the Hamming metric. In this paper we consider $b$-symbol weight distribution of such codes. Using the map $\pi_b$ in (\ref{definition.pi.b}), the problem becomes Hamming weight distribution of some $2$-weight cyclic codes over the alphabet $\F_q \times \cdots \F_q=\F_q^b$, which is not a field. We remark that two-weight irreducible cyclic codes {\it over finite fields} were characterized in \cite{SW}, and it would be interesting to obtain such a characterization over the  alphabet $\F_q^b$. We think that this would be related to Open Problem \ref{open problem} above.
\end{remark}

It remains to consider $r \le b < n$ for $C$. Using \cite[Theorem 4.2]{SOS} and Lemma \ref{lemma.determine.mu.b.is.r} we obtain a nice formula in the following.

\begin{theorem} \label{theorem.complete.b.symbol.C.b.large}
Let $a \in \F_{q^r}^*$. For $r \le b < n$ we have
\be
w_b(c(a))=n.
\nn\ee
\end{theorem}

Combining Theorem \ref{theorem.complete.b.symbol.C.b.small}  and Theorem \ref{theorem.complete.b.symbol.C.b.large} we obtain the weight $b$-symbol Hamming weight enumerators of $C$  completely.

\begin{corollary} \label{corollary.weight.enumerator}
For $2 \le b \le r-1$, the $b$-symbol Hamming weight enumerator of $C$ is
\be
A(T) = 1 + \frac{q^r-1}{2}\left(T^{u_1} + T^{u_2}\right),
\nn\ee
where
\be
u_1=
\left\{
\begin{array}{lr}
\frac{q^b-1}{N(q-1)q^{b-1}} \left( q^r - \frac{q^r+(q-1)q^{r/2}}{q}\right)
+\frac{2 \mu(b) (q-1) q^{r/2}}{Nq^b} & \mbox{if $p \equiv 1 \mod 4$}, \\ \\

\frac{q^b-1}{N(q-1)q^{b-1}} \left( q^r - \frac{q^r+(-1)^{er/2}(q-1)q^{r/2}}{q}\right)
+\frac{2 \mu(b) (-1)^{er/2} (q-1) q^{r/2}}{Nq^b} & \mbox{if $p \equiv 3 \mod 4$},
\end{array}
\right.
\nn\ee
and
\be
u_2=
\left\{
\begin{array}{lr}
\frac{q^b-1}{N(q-1)q^{b-1}} \left( q^r - \frac{q^r-(q-1)q^{r/2}}{q}\right)
-\frac{2 \mu(b) (q-1) q^{r/2}}{Nq^b}& \mbox{if $p \equiv 1 \mod 4$}, \\ \\

\frac{q^b-1}{N(q-1)q^{b-1}} \left( q^r - \frac{q^r-(-1)^{er/2}(q-1)q^{r/2}}{q}\right)
-\frac{2 \mu(b) (-1)^{er/2} (q-1) q^{r/2}}{Nq^b}& \mbox{if $p \equiv 3 \mod 4$}.
\end{array}
\right.
\nn\ee

For $r \le b < n-1$, the $b$-symbol Hamming weight enumerator of $C$ is
\be
A(T) = 1 + (q^r-1)T^n.
\nn\ee
Moreover, $C$ is an MDS $b$-symbol code when $b=r$.
\end{corollary}

\section{Proofs of the main results}
\subsection{Proof of Theorem \ref{theorem.complete.b.symbol.C.b.small}}
For $a \in \F_{q^r}$, let $\hat{c}(a) \in \F_q^{q^r-1}$ be the extended codeword of length $q^r-1$ defined as
\be
\hat{c}(a)=\left( \Tr( a \eta^{0 \cdot N}), \Tr(a \eta^{1 \cdot N}), \ldots, \Tr(a \eta^{j \cdot N}), \ldots, \Tr(a \eta^{(q^r-2)\cdot N})\right),
\nn\ee
where $0 \le j \le q^r-2$. As $\eta^{nN}=1$, we observe that
\be \label{ep1.proof.Theorem.small.b}
w(c(a))= \frac{1}{N}w(\hat{c}(a)) \;\; \mbox{and} \;\;
w_b(c(a))= \frac{1}{N}w_b(\hat{c}(a))
\ee
for each $2 \le b \le n-1$.

For $a \in \F_{q^r}$, let $Z(a) \in \C$ be defined as
\be \label{ep2.proof.Theorem.small.b}
Z(a)= \frac{1}{q} \sum_{y \in \F_q} \sum_{x \in \F_{q^r}} e^{\frac{ 2 \pi \sqrt{-1}}{p} \Tr_{q/p}\left(y \Tr_{q^r/q} (ax^N)\right)}.
\ee
For each $x \in \F_{q^r}^*$, we observe that
\be
\sum_{y \in \F_q}e^{\frac{ 2 \pi \sqrt{-1}}{p} \Tr_{q/p}\left(y \Tr_{q^r/q} (ax^N)\right)}
= \left\{
\begin{array}{rl}
q & \mbox{if $\Tr_{q^r/q} (ax^N)=0$}, \\
0 & \mbox{otherwise}.
\end{array}
\right.
\nn\ee
This implies that
\be \label{ep3.proof.Theorem.small.b}
w(\hat{c}(a))=q^r-Z(a).
\ee

Let $H \le \F_{q^r}^*$ be the multiplicative subgroup of order $\frac{q^r-1}{N}$. Using (\ref{ep2.proof.Theorem.small.b}) we obtain that
\be \label{ep4.proof.Theorem.small.b}
\begin{array}{rcl}
Z(a) & = & \dd \frac{1}{q} \left\{ q^r + \sum_{y \in \F_q^*} \sum_{x \in \F_{q^r}} e^{\frac{ 2 \pi \sqrt{-1}}{p} \Tr_{q^r/q} \left(y\Tr_{q^r/q} (ax^N)\right)} \right\} \\ \\

& = & \dd \frac{1}{q} \left\{ q^r + \sum_{y \in \F_q^*} \left[1 +  \sum_{x \in \F_{q^r}^*} e^{\frac{ 2 \pi \sqrt{-1}}{p} \Tr_{q^r/p} \left(ayx^N\right)}\right] \right\} \\ \\

& = & \dd \frac{1}{q} \left\{ q^r + q-1 + \sum_{y \in \F_q^*} \sum_{x \in \F_{q^r}^*} e^{\frac{ 2 \pi \sqrt{-1}}{p} \Tr_{q^r/p} \left(ayx^N\right)} \right\} \\ \\

& = & \dd \frac{1}{q} \left\{ q^r + q-1 + N \sum_{y \in \F_q^*} \sum_{z \in H} e^{\frac{ 2 \pi \sqrt{-1}}{p} \Tr_{q^r/p} \left(ayz\right)} \right\}
\end{array}
\ee
It is well known (see, for example, \cite[Theorem 5.15]{LN}) that
\be \label{ep5.proof.Theorem.small.b}
\sum_{
\substack{y \in \F_{q^r}^* \\ \mbox{$y$ is square}}}
e^{\frac{ 2 \pi \sqrt{-1}}{p} \Tr_{q^r/p} (y)}
= \left\{
\begin{array}{rl}
\frac{-q^{r/2}-1}{2} & \mbox{if $p \equiv 1 \mod 4$,} \\
\frac{-(-1)^{er/2} q^{r/2}-1}{2} & \mbox{if $p \equiv 3 \mod 4$,}
\end{array}
\right.
\ee
and
\be \label{ep6.proof.Theorem.small.b}
\sum_{
\substack{y \in \F_{q^r}^* \\ \mbox{$y$ is non-square}}}
e^{\frac{ 2 \pi \sqrt{-1}}{p} \Tr_{q^r/p} (y)}
= \left\{
\begin{array}{rl}
\frac{q^{r/2}-1}{2} & \mbox{if $p \equiv 1 \mod 4$,} \\
\frac{(-1)^{er/2}q^{r/2}-1}{2} & \mbox{if $p \equiv 3 \mod 4$,}
\end{array}
\right.
\ee
The following lemma is also useful.

\begin{lemma} \label{lemma1.proof.theorem.b.small}
Consider $\F_q^* \times H$ as a group with componentwise multiplicative operations. Let  $\phi$ be the map defined as
\be
\begin{array}{rcl}
\phi: \F_q^* \times H & \ra& \F_{q^r}^* \\
(y,z) & \mapsto & yz
\end{array}
\nn\ee
Then $\phi$ is a group homomorphism and $\Ker \phi=\{(h,h^{-1}): h \in H_0\}$, where $H_0 \le \F_{q^r}^*$ is the multiplicative subgroup of order $\frac{2(q-1)}{N}$.
\end{lemma}
\begin{proof}
Note that $(y,z) \in \Ker \phi$ if and only if $yz=1$. This implies that if $(y,z) \in \Ker \phi$, then $y \in \F_q^* \cap H$,
$z \in \F_q^* \cap H$ and $z=y^{-1}$. Conversely if $h \in \F_q^* \cap H$, then $(h,h^{-1}) \in \Ker \phi$. It remains to prove that $\F_q^* \cap H=H_0$.  This is equivalent to the statement that
\be \label{ep1.lemma1.proof.theorem.b.small}
\gcd\left(q-1, \frac{q^r-1}{N} \right)= \frac{2(q-1)}{N}.
\ee
Note that $\gcd\left( \frac{q^r-1}{q-1},N\right)=2$ by one of the main assumptions of this paper in Section 2. Hence we have
\be \label{ep2.lemma1.proof.theorem.b.small}
\gcd\left( \frac{q^r-1}{2(q-1)}, \frac{N}{2}\right)=1.
\ee
As $\left.\frac{N}{2} \right| \frac{q^r-1}{2}$ and $\frac{q^r-1}{2}=\frac{q^r-1}{2(q-1)} \cdot (q-1)$, using (\ref{ep2.lemma1.proof.theorem.b.small}) we conclude that
\be \label{ep3.lemma1.proof.theorem.b.small}
\left.\frac{N}{2} \right| (q-1).
\ee
Moreover $2 \left| \frac{q^r-1}{q-1}\right.$ as $r$ is even and $q$ is odd. Using (\ref{ep3.lemma1.proof.theorem.b.small}) we obtain that
\be \label{ep4.lemma1.proof.theorem.b.small}
\frac{2(q-1)}{N} \left| \gcd\left( q-1, \frac{q^r-1}{N}\right)\right..
\ee
Combining (\ref{ep2.lemma1.proof.theorem.b.small}) and (\ref{ep4.lemma1.proof.theorem.b.small}) we conclude that the statement in (\ref{ep1.lemma1.proof.theorem.b.small}) holds. This completes the proof.
\end{proof}
 In the following lemma we use the multiset notation $\{* \cdots *\}$ and the multiplicities are denoted as multiplication coefficients. For example we have $\{*x_1,x_1,x_2*\}=\{*2 \cdot x_1, x_2*\}$.

\begin{lemma} \label{lemma2.proof.theorem.b.small}
Let $a \in \F_{q^r}^*$. Using the multiset notation as above we have the followings:

\begin{itemize}
\item If $a$ is a square, then
\be
\{* ayz: y \in \F_q^*, \; z \in H *\}
=\left\{* \frac{2(q-1)}{N} \cdot x: x \in \F_{q^r}^*, \; \mbox{$x$ is a square}*\right\}.
\nn\ee
\item If $a$ is a non-square, then
\be
\{* ayz: y \in \F_q^*, \; z \in H *\}
=\left\{* \frac{2(q-1)}{N} \cdot x: x \in \F_{q^r}^*, \; \mbox{$x$ is a non-square}*\right\}.
\nn\ee

\end{itemize}

\end{lemma}
\begin{proof}
First we assume that $a=1$. The homomorphism $\phi: \F_q^* \times H \ra \F_{q^r}^*$ defined in Lemma \ref{lemma1.proof.theorem.b.small} has a kernel of cardinality $\frac{2(q-1)}{N}$ by Lemma \ref{lemma1.proof.theorem.b.small}. This implies that the image of $\phi$ has cardinality
\be \label{ep1.lemma2.proof.theorem.b.small}
\frac{(q-1) \frac{q^r-1}{N}}{\frac{2(q-1)}{N}}=\frac{q^r-1}{2}.
\ee
Moreover $\phi(y,z)$ is a square in $\F_{q^r}^*$ for any $y \in \F_q^*$ and $z \in H$ as $r$ is even and $N$ is even. Hence using (\ref{ep1.lemma2.proof.theorem.b.small}) we obtain that $\Ima \phi$ is exactly the set $\{x \in \F_{q^r}^*: x \mbox{ is a square}\}$. Counting the multiplicities we complete the proof if $a=1$.

If $a \in \F_{q^r}^* \setminus \{1\}$ is a square, then it is clear that the multiset $\{* ayz: y \in \F_q^*, \; z \in H*\}$
is the same as the multiset $\{* yz: y \in \F_q^*, \; z \in H*\}$.

If $a \in \F_{q^r}^*$ is a non-square, then we obtain that the set $\{ ayz: y \in \F_q^*, \; z \in H\}$ is equal to the set $\{ x: x \in \F_{q^r}^*, \; \mbox{$x$ is a non-square}\}$. Moreover the multiplicities are the same. This completes the proof.
\end{proof}

Assume that $a \in \F_{q^r}^*$ is a square. Using (\ref{ep4.proof.Theorem.small.b}), (\ref{ep5.proof.Theorem.small.b}) and Lemma \ref{lemma2.proof.theorem.b.small} we obtain that
\be \label{ep7.proof.Theorem.small.b}
\begin{array}{rcl}
Z(a) & = & \dd  \frac{1}{q} \left\{ q^r + q - 1 + N \frac{2(q-1)}{N}
\sum_{\substack{y \in \F_{q^r}^* \\ \mbox{$y$ is square}}}
e^{\frac{ 2 \pi \sqrt{-1}}{p} \Tr_{q^r/p} (y)} \right\} \\ \\

& = & \dd
\left\{
\begin{array}{ll}
\frac{1}{q} \left( q^r + q -1 + 2(q-1) \frac{-q^{r/2}-1}{2} \right) & \mbox{if $p \equiv 1 \mod 4$}, \\ \\
\frac{1}{q} \left( q^r + q -1 + 2(q-1) \frac{q^{r/2}-1}{2} \right) & \mbox{if $p \equiv 3 \mod 4$},
\end{array}
\right. \\ \\
& = & \dd
\left\{
\begin{array}{ll}
\frac{q^r -(q-1)q^{r/2}}{q} & \mbox{if $p \equiv 1 \mod 4$}, \\ \\
\frac{q^r - (-1)^{er/2}(q-1)q^{r/2}}{q} & \mbox{if $p \equiv 3 \mod 4$},
\end{array}
\right.
\end{array}
\ee

Next we assume that $a \in \F_{q^r}^*$ is a non-square. Using (\ref{ep4.proof.Theorem.small.b}), (\ref{ep6.proof.Theorem.small.b}) and Lemma \ref{lemma2.proof.theorem.b.small} similarly we obtain that
\be \label{ep8.proof.Theorem.small.b}
\begin{array}{rcl}
Z(a) & = & \dd
\left\{
\begin{array}{ll}
\frac{q^r +(q-1)q^{r/2}}{q} & \mbox{if $p \equiv 1 \mod 4$}, \\ \\
\frac{q^r + (-1)^{er/2}(q-1)q^{r/2}}{q} & \mbox{if $p \equiv 3 \mod 4$},
\end{array}
\right.
\end{array}
\ee

Recall that the set $\mathcal{P}(b)$ is defined in Definition \ref{definition.set.P.b}. An important result obtained in \cite{SOS} using algebraic curves over finite fields is that
\be \label{ep9.proof.Theorem.small.b}
w_b(\hat{c}(a))= \frac{1}{q^{b-1}} \sum_{\theta \in \mathcal{P}(b)} w \left( \hat{c}(\theta a)\right).
\ee

Recall that $\mu(b)$ is the integer defined in Definition \ref{definition.mu.b}.

Now assume that $p \equiv 1 \mod 4$ and $a$ is a square in $\F_{q^r}^*$. Combining (\ref{ep3.proof.Theorem.small.b}), (\ref{ep7.proof.Theorem.small.b}), (\ref{ep8.proof.Theorem.small.b}) and (\ref{ep9.proof.Theorem.small.b}) we obtain that
\be \label{ep10.proof.Theorem.small.b}
\begin{array}{rcl}
w_b\left(\hat{c}(a)\right) & = & \dd \frac{1}{q^{b-1}} \left( \mu(b) \left[ q^r - \frac{q^r - (q-1)q^{r/2}}{q}\right]\right) \\ \\
& & + \frac{1}{q^{b-1}} \left( \left( \frac{q^b-1}{q-1}-\mu(b) \right)  \left[ q^r - \frac{q^r + (q-1)q^{r/2}}{q}\right]\right) \\ \\
& = & \dd  \frac{q^b-1}{(q-1)q^{b-1}} \left( q^r - \frac{q^r + (q-1)q^{r/2}}{q}\right)  + \frac{2 \mu(b) (q-1) q^{r/2}}{q^b}.
\end{array}
\ee

Next assume that $p \equiv 1 \mod 4$ and $a$ is a non-square in $\F_{q^r}^*$. Combining (\ref{ep3.proof.Theorem.small.b}), (\ref{ep7.proof.Theorem.small.b}), (\ref{ep8.proof.Theorem.small.b}) and (\ref{ep9.proof.Theorem.small.b}) we obtain that
\be \label{ep11.proof.Theorem.small.b}
\begin{array}{rcl}
w_b\left(\hat{c}(a)\right) & = & \dd \frac{1}{q^{b-1}} \left( \mu(b) \left[ q^r - \frac{q^r + (q-1)q^{r/2}}{q}\right]\right) \\ \\
& & + \frac{1}{q^{b-1}} \left( \left( \frac{q^b-1}{q-1}-\mu(b) \right)  \left[ q^r - \frac{q^r - (q-1)q^{r/2}}{q}\right]\right) \\ \\
& = & \dd  \frac{q^b-1}{(q-1)q^{b-1}} \left( q^r - \frac{q^r - (q-1)q^{r/2}}{q}\right)  - \frac{2 \mu(b) (q-1) q^{r/2}}{q^b}.
\end{array}
\ee

Again next assume that $p \equiv 3 \mod 4$ and $a$ is a square in $\F_{q^r}^*$. Combining (\ref{ep3.proof.Theorem.small.b}), (\ref{ep7.proof.Theorem.small.b}), (\ref{ep8.proof.Theorem.small.b}) and (\ref{ep9.proof.Theorem.small.b}) we obtain that
\be \label{ep12.proof.Theorem.small.b}
\begin{array}{rcl}
w_b\left(\hat{c}(a)\right) & = & \dd \frac{1}{q^{b-1}} \left( \mu(b) \left[ q^r - \frac{q^r - (-1)^{er/2} (q-1)q^{r/2}}{q}\right]\right) \\ \\
& & + \frac{1}{q^{b-1}} \left( \left( \frac{q^b-1}{q-1}-\mu(b) \right)  \left[ q^r - \frac{q^r + (-1)^{er/2} (q-1)q^{r/2}}{q}\right]\right) \\ \\
& = & \dd  \frac{q^b-1}{(q-1)q^{b-1}} \left( q^r - \frac{q^r + (-1)^{er/2} (q-1)q^{r/2}}{q}\right)  + \frac{2 \mu(b) (-1)^{er/2} (q-1) q^{r/2}}{q^b}.
\end{array}
\ee

Finally assume that $p \equiv 3 \mod 4$ and $a$ is a non-square in $\F_{q^r}^*$. Combining (\ref{ep3.proof.Theorem.small.b}), (\ref{ep7.proof.Theorem.small.b}), (\ref{ep8.proof.Theorem.small.b}) and (\ref{ep9.proof.Theorem.small.b}) we obtain that
\be \label{ep13.proof.Theorem.small.b}
\begin{array}{rcl}
w_b\left(\hat{c}(a)\right) & = & \dd \frac{1}{q^{b-1}} \left( \mu(b) \left[ q^r - \frac{q^r + (-1)^{er/2}(q-1)q^{r/2}}{q}\right]\right) \\ \\
& & + \frac{1}{q^{b-1}} \left( \left( \frac{q^b-1}{q-1}-\mu(b) \right)  \left[ q^r - \frac{q^r - (-1)^{er/2} (q-1)q^{r/2}}{q}\right]\right) \\ \\
& = & \dd  \frac{q^b-1}{(q-1)q^{b-1}} \left( q^r - \frac{q^r - (-1)^{er/2} (q-1)q^{r/2}}{q}\right)  - \frac{2 \mu(b) (-1)^{er/2} (q-1) q^{r/2}}{q^b}.
\end{array}
\ee

Using (\ref{ep1.proof.Theorem.small.b}), (\ref{ep10.proof.Theorem.small.b}), (\ref{ep11.proof.Theorem.small.b}),
(\ref{ep12.proof.Theorem.small.b}) and (\ref{ep13.proof.Theorem.small.b}) we complete the proof of Theorem \ref{theorem.complete.b.symbol.C.b.small}.
\subsection{Proof of Theorem \ref{theorem.complete.b.symbol.C.b.large}}
If $r < b \le n-1$, then using the proof of \cite[Theorem 4.2]{SOS} we obtain that $w_b(c(a))=w_r(c(a))$. It remains to prove Theorem \ref{theorem.complete.b.symbol.C.b.large} for $b=r$.

Assume that $p \equiv 1 \mod 4$ and $a$ is a square in $\F_{q^r}^*$. Then using Lemma \ref{lemma.determine.mu.b.is.r} and Theorem \ref{theorem.complete.b.symbol.C.b.small} we obtain that $ \mu(r)=\frac{q^r-1}{2(q-1)}$ and hence
\be
\begin{array}{rcl}
w_r(c(a)) & = & \dd \frac{q^r-1}{N(q-1)q^{r-1}} \left( q^r - \frac{q^r + (q-1)q^{r/2}}{q}\right) \dd + \frac{(q^r-1)q^{r/2}}{Nq^r} \\ \\
 & = & \dd \frac{q^r-1}{N(q-1)q^{r-1}} \left( q^r - q^{r-1}\right) \dd + \frac{q^r-1}{Nq^r} \left( -q^{r/2} + q^{r/2} \right) \\ \\
 & = & \dd \frac{\left(q^r-1\right)q^{r-1}}{Nq^{r-1}} = \frac{q^r-1}{N}=n.
\end{array}
\nn\ee
The proof is similar for the case that $p \equiv 1 \mod 4$ and $a$ is a non-square.

Next assume that $p \equiv 3 \mod 4$ and $a$ is a square in $\F_{q^r}^*$. Then similarly  using Lemma \ref{lemma.determine.mu.b.is.r} and Theorem \ref{theorem.complete.b.symbol.C.b.small} we obtain that
\be
\begin{array}{rcl}
w_r(c(a)) & = & \dd \frac{q^r-1}{N(q-1)q^{r-1}} \left( q^r - \frac{q^r + (-1)^{er/2} (q-1)q^{r/2}}{q}\right) \dd + \frac{(-1)^{er/2}(q^r-1)q^{r/2}}{Nq^r} \\ \\
 & = & \dd \frac{q^r-1}{N(q-1)q^{r-1}} \left( q^r - q^{r-1}\right)  \dd + \frac{(-1)^{er/2}q^r-1}{Nq^r} \left( -q^{r/2} + q^{r/2} \right) \\ \\
 & = & \dd \frac{\left(q^r-1\right)q^{r-1}}{Nq^{r-1}} = \frac{q^r-1}{N}=n.
\end{array}
\nn\ee
The proof is similar for the case that $p \equiv 3 \mod 4$ and $a$ is a non-square. From the definition of the MDS $b$-symbol codes, it is easy to check that $C$ is MDS when $b=r$.
This completes the proof of Theorem \ref{theorem.complete.b.symbol.C.b.large}.
\begin{example}
Let $\gamma$ be a primitive element of $\F_{81}$ with $\gamma^4+2\gamma^3+2=0$.

According to Table 1 of \cite{Ding}, the Hamming weight enumerator of $C$ is $1+A^{(1)}_{24}T^{24}+A^{(1)}_{30}T^{30}
=1+40T^{24}+40T^{30}$ when $q=3, r=4, N=2$ and $\gcd(\frac{q^r-1}{q-1},N)=2$. 
 When $b=2$, $|\sP(2)|=q+1=4$.
Let $\sP(2)=\{1,\gamma^2,1+\gamma^2, 1+2\gamma^2\}$. If $a\in C_i^{(2,r)}$, then $\gamma^2a, (1+\gamma^2)a=
\gamma^{58}a\in C_i^{(2,r)},
(1+2\gamma^2)a=\gamma^{65}a\in C_{i+1 ({\rm mod}~ 2)}^{(2,r)}$, where $i=0,1$.
According to the equality (\ref{ep9.proof.Theorem.small.b}), we have
\begin{small}
\begin{eqnarray*}
  w_2({c}{(a)}) &=& \frac{1}{3}[w_1({c}{(a)})
  +w_1({c}{(\gamma^2a)})
  +w_1({c}{((1+\gamma^2)a)})
  +w_1({c}{((1+2\gamma^2)a)})]\\
  ~ &=& \left \{
               \begin{aligned}
               &  \frac{1}{3}(24+24+24+30)=34
  &\quad a\in C_0^{(2,r)}, &  \\
                &  \frac{1}{3}(30+30+30+24)=38 &a\in C_1^{(2,r)}. &
              \end{aligned}
\right.
\end{eqnarray*}
\end{small}
Therefore the $2$-symbol Hamming weight enumerator of $C$ is $1+A_{34}^{(2)}T^{34}+A_{38}^{(2)}T^{38}=1+|C_0^{(2,r)}|T^{34}+
|C_1^{(2,r)}|T^{38}=1+40T^{34}+40T^{38}$.

Combining Example \ref{ep2.2} and Corollary \ref{corollary.weight.enumerator}, we have $\mu(2)=3$, $u_1=38$ and $u_2=34$, which is the same as above result.
The above result also verified by Magma experiment.
\end{example}

\section*{Acknowledgement}
This research is supported by National Natural Science Foundation of China (12071001,\\61672036), Excellent Youth Foundation of Natural Science Foundation of Anhui Province (1808085J20) and the Academic Fund for Outstanding Talents in Universities (gxbjZD03).

The authors are grateful to Prof. Tor Helleseth for helpful discussions and
suggestions.

\section{appendix}
Using algebraic curves, Shi {\it et al.} \cite{SOS} proved that the equality (\ref{ep9.proof.Theorem.small.b}) holds. Here we give another proof which using exponent sum.

Let $C$ be a cyclic irreducible code of period $n$ and dimension $k$ with zero $\eta=\gamma^{N}$ where $\gamma^N$ is a primitive element in $\F_{q^r}$ and $N=\frac{q^k-1}{n}$. Let $I=\{1,\eta,\ldots,\eta^{n-1}\}.$ Recall that $C=\{c(a)=\left( \Tr( a \eta^{0 \cdot N}), \Tr(a \eta^{1 \cdot N}), \ldots, \Tr(a \eta^{j \cdot N}), \ldots, \Tr(a \eta^{(n-1)\cdot N})\right)\}$ and the Hamming weight of a codeword $w_1(c(a))$ is $$n-w_1(c(a))=\sum_{x\in I}\frac{1}{q}\sum_{y\in\F_q}\chi(yax),$$
where $\chi$ is an additive character of $\F_{q^r}$.
To find the $b$-symbol Hamming weight of $c(a)$, then one needs to count how often $b$ consecutive positions $(\Tr(ax),\Tr(a\eta x),\cdots,\Tr(s\eta^b x))$ are all $0$. This implies
\begin{eqnarray*}
  n-w_b(c(a)) &=& \sum_{x\in I}\frac{1}{q^b}\sum_{y_1,\ldots,y_b\in\F_q}\chi((y_1+y_2\eta+\cdots+y_{b-1}\eta^{b-1})ax) \\
  ~ &=& \frac{1}{q^{b-1}}\sum_{x\in I}\frac{1}{q}\sum_{u\in V}\chi(uax),
\end{eqnarray*}
where $V=<1,\eta,\cdots,\eta^{b-1}>.$ Let $\bigcup_{i=0}^{\frac{q^b-1}{q-1}}\theta_i\F_q^*\cup0$ then
\begin{eqnarray*}
  n-w_b(c(a)) &=& \frac{1}{q^{b-1}}\sum_{i=1}^{\frac{q^b-1}{q-1}}\sum_{x\in I}\frac{1}{q}\sum_{y\in\F_q^*}\chi(y\theta_iax)+\frac{n}{q^b} \\
  ~ &=& \frac{1}{q^{b-1}}\sum_{i=1}^{\frac{q^b-1}{q-1}}(n-w_1(c(\theta_ia)))-\frac{n(q^b-1)}{q^b(q-1)}+\frac{n}{q^b} \\
  ~ &=& \frac{n(q^b-1)}{q^{b-1}(q-1)}-\frac{1}{q^{b-1}}\sum_{i=1}^{\frac{q^b-1}{q-1}}w_1(c(\theta_ia))-\frac{n(q^b-q)}{q^b(q-1)} \\
  ~ &=& n-\frac{1}{q^{b-1}}\sum_{i=1}^{\frac{q^b-1}{q-1}}w_1(c(\theta_ia)).
\end{eqnarray*}
Hence, $w_b(c(a))=\frac{1}{q^{b-1}}\sum_{i=1}^{\frac{q^b-1}{q-1}}w_1(c(\theta_ia)).$


\end{document}